\documentclass[sigconf]{acmart}

\usepackage{amsmath,amsfonts} 
\usepackage{algorithmic}
\usepackage{graphicx}
\usepackage{textcomp}
\usepackage{xcolor}
\usepackage[bb=boondox]{mathalfa}
\usepackage{amsthm}
\usepackage{subfigure}
\usepackage{caption}
\usepackage{subcaption}
\usepackage{comment}

\theoremstyle{plain}

\newtheorem*{lem*}{Lemma}

\newtheorem*{theorem*}{Theorem}
\theoremstyle{definition}
\newtheorem{definition}{Definition}

\AtBeginDocument{%
  }

\setcopyright{acmcopyright}
\copyrightyear{2018}
\acmYear{2018}
\acmDOI{XXXXXXX.XXXXXXX}

\acmConference[Conference acronym 'XX]{Make sure to enter the correct
  conference title from your rights confirmation emai}{June 03--05,
  2018}{Woodstock, NY}
\acmPrice{15.00}
\acmISBN{978-1-4503-XXXX-X/18/06}

\begin{document}

\title{Functional Type Expressions of Sequential Circuits \\ with the Notion of Referring Forms}

\author{Shunji Nishimura}
\email{s-nishimura@oita-ct.ac.jp}
\orcid{0000-0001-6600-5136}
\affiliation{%
  \institution{National Institute of Technology, Oita College}
  \streetaddress{Maki 1666}
  \city{Oita City}
  \state{Oita Prefecture}
  \country{Japan}
  \postcode{870-0152}
}


\begin{abstract}
This paper introduces the notion of referring forms as a new metric for analyzing sequential circuits from a functional perspective.
Sequential circuits are modeled as causal stream functions, the outputs of which depend solely on the past and current inputs.
Referring forms are defined based on the type expressions of functions and represent how a circuit refers to past inputs.
The key contribution of this study is identifying a universal property in multiple clock domain circuits using referring forms.
This theoretical framework is expected to enhance the comprehension and analysis of sequential circuits. 
\end{abstract}

\ccsdesc[300]{Theory of computation~Streaming models}
\ccsdesc[300]{Theory of computation~Timed and hybrid models}

\keywords{sequential circuit, stream function, causal function}


\maketitle

\section{Introduction}

The primary expectation of digital circuits is functionality.
We are concerned with how the circuit produces output when a sequential series of input values is provided, from an external perspective.
Although the internal structure of a circuit that realizes this functionality is important, it can be considered as secondary.
Further theoretical studies on the functionality of these circuits are warranted. 
We discuss digital circuits from a theoretical and functional perspective, without focusing on practical issues such as circuit design for a particular application.
In this study, sequential circuits are viewed as causal functions, the outputs of which depend solely on past and current inputs and not on future inputs.
We adopt type expressions to investigate these functions.

A previous study \cite{ICCSS2023} adopted a similar approach, but with three profound differences.
(1) Regarding background time, the previous study allowed an arbitrary partially ordered set as time, whereas this study considers stream functions; that is, only natural (ordinal) numbers are regarded as time, and this difference simplifies type expressions.
(2) This study proposes the notion of referring forms, which appeared only implicitly in the previous paper.
(3) The notion of time preservation was also proposed in the previous study, and multiple clock domain circuits were classified as time preserving without proof.
In this study, the proof is provided.

The contribution of this work is that it proposes a new metric to investigate sequential circuits, namely referring forms, accompanied by the revelation of a universal property of multiple clock domain circuits.
We expect that these results will enhance the understanding of sequential circuits.

\section{Preliminaries}
\subsection{Types of Causal Stream Functions}
From a theoretical perspective, sequential circuits are described using Mealy machines (or Mealy automata) \cite{mealy1955method,holcombe2004algebraic}.
For arbitrary sets $I$ and $O$, a Mealy machine $(S,\psi)$ with input $I$ and output $O$ consists of a set $S$ of states and transition function $\psi : S\to (O\times S)^I$. When a previous state $s\in S$ and current input $i\in I$ are provided, a tuple of the current output and state $\psi\,s\,i \in O\times S$ is derived. For a given initial state, the Mealy machine $(S,\psi)$ behaves as a causal stream function of type
$I^{\mathbb{N}_1} \to O^{\mathbb{N}_1}
$,
where $\mathbb{N}_1$ denotes the natural numbers $\{1,2,\cdots\}$, excluding zero, as ordinal numbers.
(We also use $\mathbb{N}_0$ for $\{0,1,\cdots\}$ later.)
Domain $I^{\mathbb{N}_1}$ and codomain $O^{\mathbb{N}_1}$ represent infinite input and output streams over time; that is, the input and output signals, respectively. Causal means that the current output depends only on past and current inputs and not on future inputs.

This function is constructed from the transition function $\psi$ of the Mealy machine, as follows: Given an initial state $s_0$ and a finite input stream $(i_1,i_2,\cdots,i_n)\in I^n $, the output stream $(o_1,o_2,\cdots,o_n)\in O^n$ is determined by $(o_j,s_j)=\psi\,s_{j-1}\,i_j\;(j\in\mathbb{N}_1)$.
Conversely, \cite{rutten2006algebraic} proved that an arbitrary causal stream function $f: I^{\mathbb{N}_1} \to O^{\mathbb{N}_1}$ can be represented using the corresponding Mealy machine.
Therefore, we investigate the causal stream functions $I^{\mathbb{N}_1} \to O^{\mathbb{N}_1}$ instead of Mealy machines.

    We consider a causal stream function $f\,:\,I^{\mathbb{N}_1} \to O^{\mathbb{N}_1}$.
Because it is causal, the $n$-th output $o_n$ depends on an $n$-length input stream $(i_1,i_2,\cdots,i_n)$, and $f$ can be expressed by a family of functions $\{\,f_n\,:\,I^n\to  O\,\}_{n \in \mathbb{N}_1}$.
We denote this function family as
\begin{align}
I^+\to O
\label{eq:IplusToO}
\end{align}
to simplify the later expressions.

We must consider the product type of the input to develop our theory further.
When the input $I$ of the causal stream functions (\ref{eq:IplusToO}) is converted into the product $I_A\times I_B$, it becomes $(I_A\times I_B)^+ \cong I_A^{\;+} \times I_B^{\;+}$ and the type is isomorphic to

\begin{align}
I_A^{\;+}\to I_B^{\;+}\to O
\label{eq:IaIbO}
\end{align}
by Currying. Note that the two pluses above must be interpreted as the same number; that is, $\{\,I_A^{\;\,n} \to I_B^{\;\,n} \to O\,\}_{n \in \mathbb{N}_1}$ in the expression of function families, which differs from the usual definition $A^+=A+(A\times A)+\cdots+A^\mathbb{N}$.

\subsection{Dependent Types}
We also require notions of the dependent types \cite{martin1975intuitionistic,chlipala2022certified} for sets and elements instead of types and terms.
For given $A \in Set$ and $\alpha: A \to Set$, a subset $(a: A) \times \alpha\,a \,\subset\, A \times {\displaystyle \sum_{x\in A}}\,\alpha\,x$ is defined as
\begin{align}
(x,y) \in (a: A) \times \alpha\,a \enspace\;\text{iff}\;\enspace x\in A \;\text{and}\; y\in \alpha\,x,
\end{align}
where $(a : A)$ represents a domain $A$ and an arbitrary $a$ is obtained from $A$ for later expressions.
In this study, $(a : A)$ denotes a domain with dependent types, as explained above, and $a\in A$ denotes the usual proposition.

\section{Referring Forms}

\subsection{Definition}

For functions with two-part inputs, such as (\ref{eq:IaIbO}), we introduce the notion that a given first input restricts the second input domain.
To express this notion, we use dependent types and consider
\begin{align}
\label{eq:abstractPlus}
\begin{gathered}
f:(\sigma : C^+) \to \phi\, \sigma \to O \\
\phi : C^+ \to \mathcal{P}(I^+), 
\end{gathered}
\end{align}
where $\mathcal{P}$ denotes the power sets. We use the control input $C$ and data input $I$ instead of $I_A$ and $I_B$ in (\ref{eq:IaIbO}).
The second domain of the upper expression in (\ref{eq:abstractPlus}) is not the entire data input $I^+$, but only part of it, as if it were filtered based on the control input $\sigma$.
For example, for given $\phi_n : C^n \to \mathcal{P}(I^n)$ and $\sigma_n = (c_1,c_2,\cdots,c_n)$, $\phi_n\,\sigma_n := I\times\cdots\times I$ means that $f_n\,\sigma_n$ can refer to all past inputs, whereas $\phi_n\,\sigma_n:=\emptyset\times\cdots\times\emptyset\times I$ means that $f_n\,\sigma_n$ can refer to the current input only.
In addition, we assume that a given first value $\sigma : C^+$ of $f$ does not influence the final function $\phi\, \sigma \to O$; that is, $\phi\,\sigma = \phi\,\sigma'$ implies $f\,\sigma = f\, \sigma'$.

Our model can be made more precise by focusing on the first and final time steps.
In the first time step, no past is referred to; thus, $C^+$ must be modified to $C^*$. In the final time step, the current input is always referred to, because we consider Mealy machines.
Finally, we obtain the following definition, originally proposed in \cite{ICCSS2023}:

\begin{definition} (Domain restriction.) \label{def:DR}
For given control input $C$ and data input $I$, Mealy machines as causal stream functions with domain restrictions, as follows:
\begin{align}
\label{eq:DR}
\begin{gathered}
f:(\sigma : C^*) \to \phi\, \sigma \times I \to O \\
\phi : C^* \to \mathcal{P}(I^*),
\end{gathered}
\end{align}
and $\phi$ is known as a restriction map.
\end{definition}

The current input $I$ is always provided, and the restriction map $\phi$ considers only past inputs and not current inputs.
A restriction map $\phi$ comprehends how the circuit refers to the past with the passage of time.

The notion of domain restriction (\ref{eq:DR}) contains the control signal $\sigma \in C^*$; however, to study the reference method, we only require the image $\phi\,\sigma$.
We omit the control signal $\sigma$ from the expression of the domain restriction (\ref{eq:DR}), which helps us to observe how the circuit refers to past inputs easily.

\begin{definition}(Referring forms.)
For a causal stream function that can be expressed as
\begin{align}
(\rho : \mathcal{P}(I^*)) \to \rho \times I \to O,
\label{eq:RF}
\end{align}
$\rho$ is the \textit{referring form}.
\end{definition}
Roughly, referring forms represent a memory element of the target circuit and a given control signal to the memory element; they provide the manner in which the circuit refers to past inputs.
When (\ref{eq:RF}) is precisely expressed as a family of functions,
\begin{align}
\{(\rho : \mathcal{P}(I^n)) \to \rho \times I \to O \}_{n \in \mathbb{N_0}},
\end{align}
and the referring form is $\{\rho_n : \mathcal{P}(I^n) \}_{n\in\mathbb{N}_0}$.
The set $\{\rho_n\}_{n\in\mathbb{N}_0}$ can be regarded as $\rho : (n : \mathbb{N}_0) \to \mathcal{P}(I^n)$ and approximately $\rho : \mathbb{N}_0 \to \mathcal{P}(I^{\mathbb{N}_0})$.

\subsection{Examples}

\begin{figure*}
    \centering
    \begin{minipage}{0.29\textwidth}
        \centering
        \includegraphics[width=0.9\linewidth]{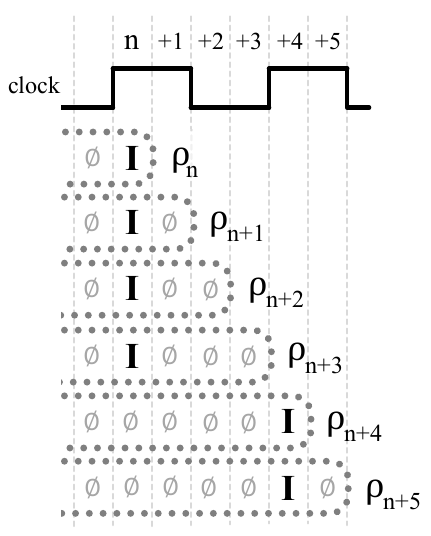}
        \caption{Referring form of D-FF.}
        \label{fig:RefDFF}
   \end{minipage}\hfill
   \begin{minipage}{0.64\textwidth}
	   \subfigure[Block diagram]{
        \includegraphics[clip, width=0.45\columnwidth]{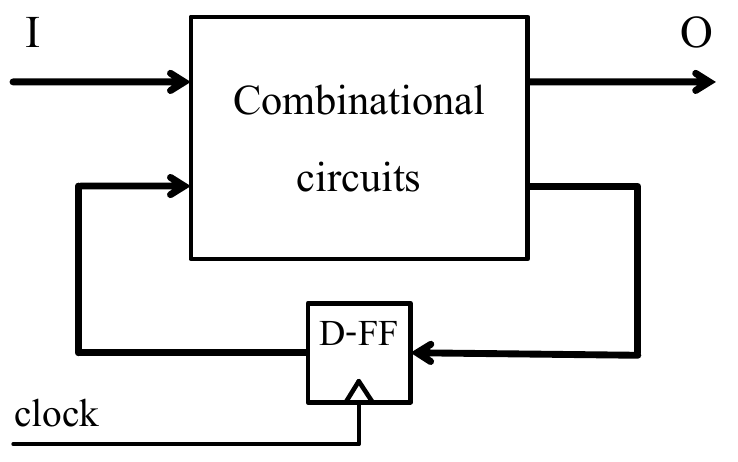}
        }
      \subfigure[Referring form]{
		  \includegraphics[clip, width=0.5\columnwidth]{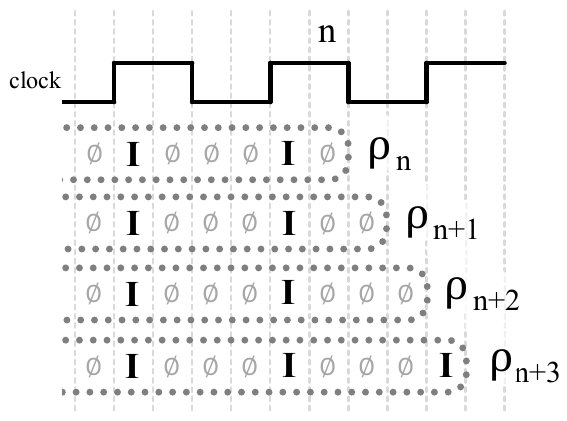}
        }
      \caption{Synchronous circuits.}
      \label{fig:Sync}
    \end{minipage}
    \begin{minipage}{0.34\textwidth}
        \centering
        \includegraphics[width=0.8\linewidth]{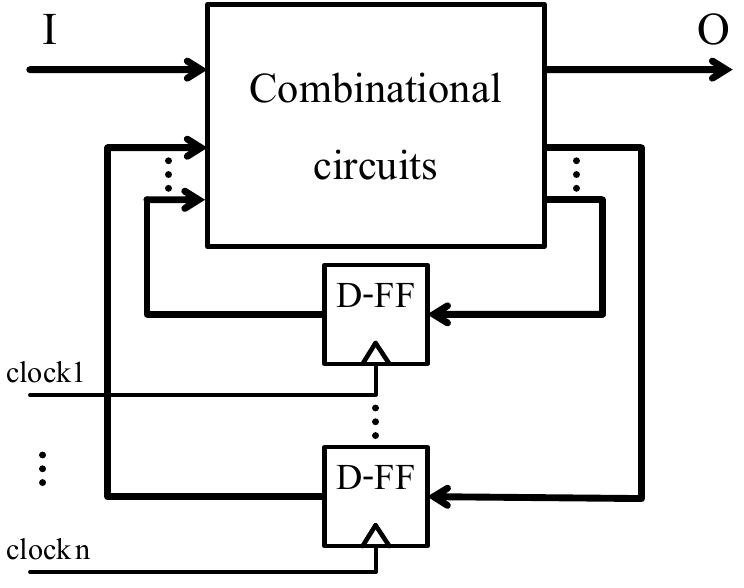}
        \caption{Multiple clock domain circuits.}
        \label{fig:DiagMCD}
    \end{minipage}
    \begin{minipage}{0.64\textwidth}
        \subfigure[Block diagram]{
            \includegraphics[width=0.50\linewidth]{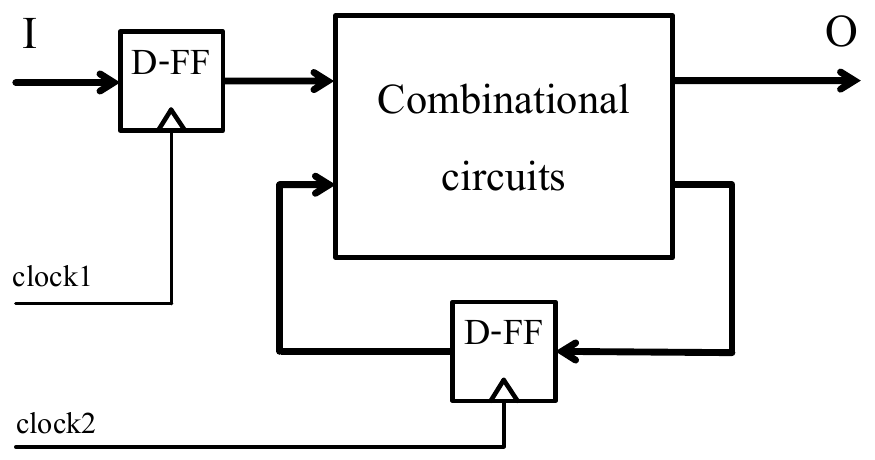}
            }
        \subfigure[Referring form]{
            \includegraphics[width=0.45\linewidth]{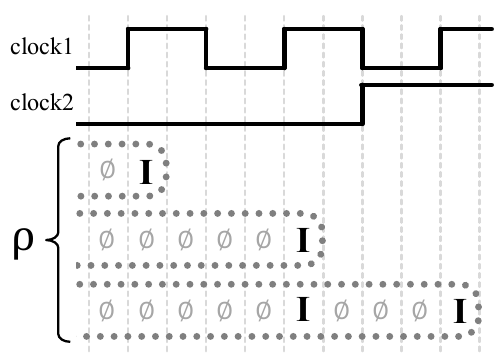}
            }
        \caption{Two-clock domain circuit.}
        \label{fig:2CDin}
    \end{minipage}
\end{figure*}

We present an example of referring forms for a D-flip-flop (D-FF) as a sequential circuit, with data input $I$, output $O$, and a clock.
Because the D-FF holds the input value at the latest edge of the clock, the referring form $\rho$ behaves as shown in Fig. \ref{fig:RefDFF}; during time $n,\cdots,n+3$, $\rho$ provides the input value at $n$, following which $\rho$ provides the input value at $n+4$.
An excerpt from the overall type expressions is as follows:
\begin{align*}
\cdots \times\emptyset\times I \enspace\enspace\enspace\,&\,\qquad\quad (\times I) \to O \;(\text{at}\,n+1)\\
\cdots \times\emptyset\times I \times\emptyset&\qquad\quad \; (\times I) \to O \;(\text{at}\,n+2)\\
\cdots \times\emptyset\times I \times\emptyset&\times\emptyset\quad\;\; (\times I) \to O \;(\text{at}\,n+3),
\end{align*}
where $(\times I)$ means that the current input is not referred to by the D-FF.
Note that the referring form shown in Fig. \ref{fig:RefDFF} is regarded as $\mathcal{P}(I)^*$ because $\mathcal{P}(I)^*\cong\mathcal{P}(I^*)$ and may be intuitively comprehensible.
Henceforth, this isomorphism is sometimes used implicitly. 

Owing to the abstractness of types, we can express a class of circuits such as synchronous circuits, as shown in Fig. \ref{fig:Sync}(a), in an expression; the referring form is shown in (b).
The input $I$ is held and can be referred to at every edge of the clock.
Unlike D-FFs, synchronous circuits can potentially refer to all past inputs at the clock edges through their feedback loop.

Next, we consider multiple clock domain circuits, as shown in Fig. \ref{fig:DiagMCD}. With the general description, their referring form becomes the same that of as Fig. \ref{fig:Sync}(b) for synchronous circuits because 
every FF may hold data from the input $I$ and outputs from all of the FFs.
In fact, when multiple clock domains are required, multiple parts should independently have different clock domains; for instance, Fig. \ref{fig:2CDin}(a) has $clock 1$ to latch input $I$, followed by the general $clock 2$ domain circuit.
The referring form in Fig. \ref{fig:2CDin}(b) indicates that the latest input $I$ latched by $clock 1$ is delivered to the $clock 2$ domain when $clock 2$ ticks.

\section{Time Preservation}

In the previous section, each referring form was studied individually. In this section, we investigate the universal properties of a set of referring forms representing a particular circuit.
For a set of referring forms $R$, each referring form $\rho\in R$ is $\rho : \mathbb{N}_0 \to \mathcal{P}(I^{\mathbb{N}_0})$, and we denote the unified image $\bigcup\limits_{\rho\in R} \rho (\mathbb{N}_0)$ as $R(\mathbb{N}_0)$.
As $\mathbb{N}_0$ has an order, we can consider whether the unified image has an order preserved by each $\rho\in R$.

\begin{definition} (Time preservation.)
A set of referring forms $R$ is \textit{time preserving} when image $R(\mathbb{N}_0)$ has a partial order and each $\rho\in R$ becomes order preserving.
\label{def:TimePreserving}
\end{definition}
Time preservation can be interpreted as the image $R(\mathbb{N}_0)$ that maintains the concept of time in the original environment.

We focus on multiple clock domain circuits, as shown in Fig. \ref{fig:DiagMCD}, and study their referring forms in detail.
To generate the output $O_m$ at $m\in\mathbb{N}_0$, the circuits can refer to the previous input $I_{m-1}$ if it has been latched by some FFs, as well as to the earlier input $I_{m-2}$ if it has been latched at $m-2$ and held at $m-1$.
This approach can be described using a pseudo-directed graph, as shown in Fig. \ref{fig:GraphMCDgeneral}, which allows us to determine which inputs can be referred to each output by tracking arrows.
The nodes $FF_i (i=1,\cdots,n)$ represent D-FFs, which connect only one of their two input arrows to the output arrow depending on their attribute latch or hold, as with a multiplexer.
The circular nodes represent combinational circuits, and in this context, we consider their connections rather than their logic.
When we consider these nodes as all input arrows connected to all output arrows, the subject circuits become abstract multiple clock domain circuits.
However, when we consider that these nodes are one of every possible input-output connection, this corresponds to considering arbitrary circuits in the form of Fig. \ref{fig:DiagMCD}, and we take the latter stance.
Thus, the pseudo-graph in Fig. \ref{fig:GraphMCDgeneral} describes referring forms of multiple clock domain circuits (Fig. \ref{fig:DiagMCD}).

\begin{figure*}
    \centering
    \begin{minipage}{0.45\textwidth}
    \includegraphics[width=0.8\linewidth]{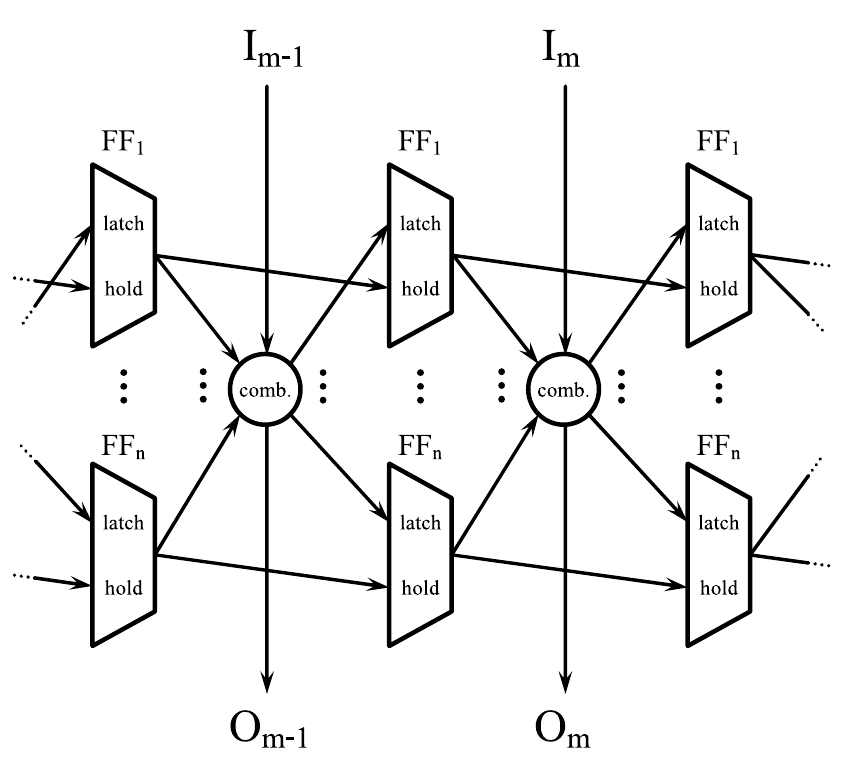}
    \caption{Pseudo-graph for referring forms of multiple clock domain circuits (Fig. \ref{fig:DiagMCD}).}
    \label{fig:GraphMCDgeneral}
\end{minipage}
\begin{minipage}{0.45\textwidth}
    \centering
    \includegraphics[width=0.8\linewidth]{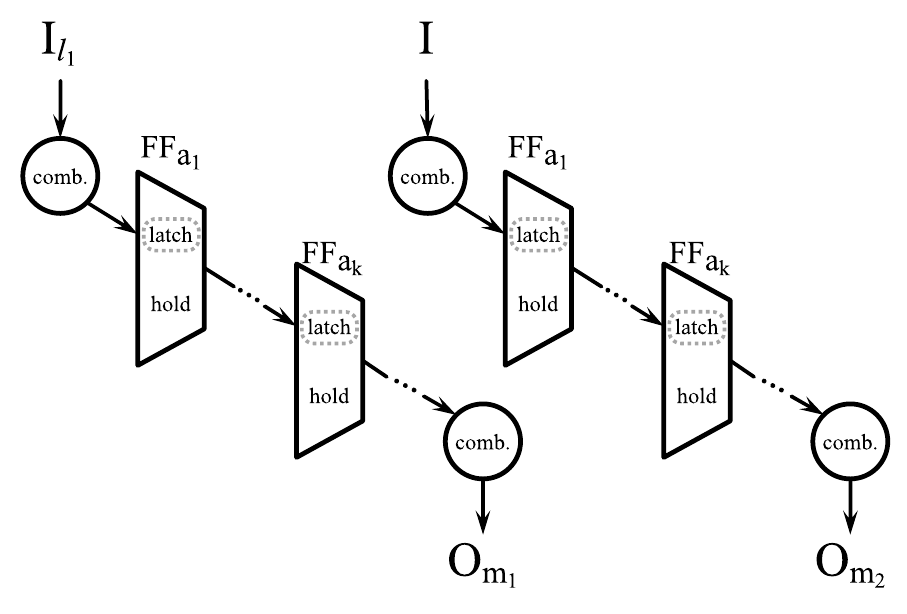}
    \caption{Pseudo-graph to prove the theory.}
    \label{fig:GraphProof}
\end{minipage}
    \subfigure[Block diagram]{
        \includegraphics[width=0.25\linewidth]{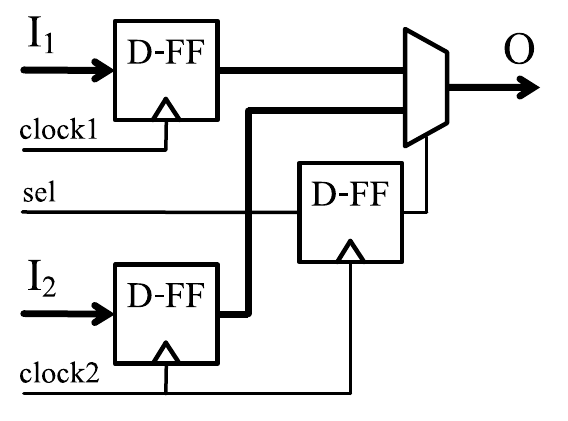}
        }
    \subfigure[Referring form $\rho$]{
        \includegraphics[width=0.25\linewidth]{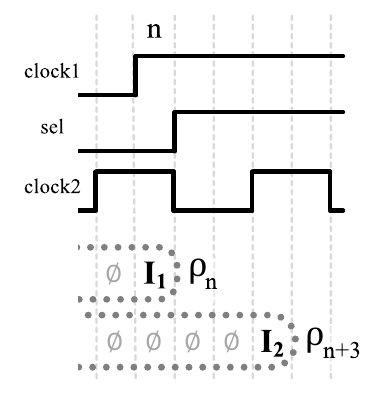}
        }
    \subfigure[Referring form $\rho'$]{
        \includegraphics[width=0.35\linewidth]{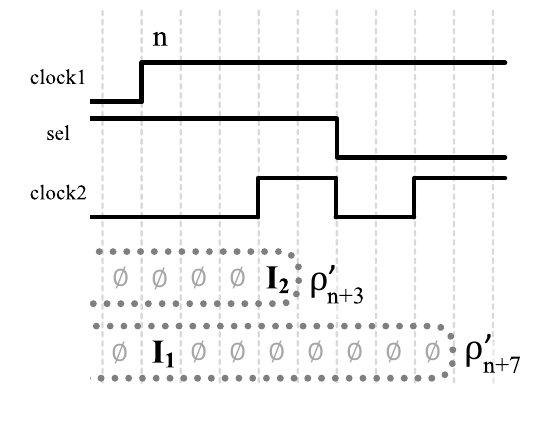}
    }
    \caption{Selective memory circuit.}
    \label{fig:2mem}
\end{figure*}

\begin{lem*}
Let $\rho$ be a referring form of multiple clock domain circuits. If $\rho_{m_1}$ contains $I$ at $l_1$, for an arbitrary $m_2 \ge m_1$, there exists $l_2 \ge l_1$ and $\rho_{m_2}$ contains $I$ at $l_2$.
\label{lem:the}
\end{lem*}
\begin{proof}
We can prove this by contradiction.
Assume that \\
$(\rho_{m_2})_{l_1},\cdots,(\rho_{m_2})_{m_2}$ are $\emptyset$.
According to $(\rho_{m_1})_{l_1} = I$, we assume that the path on the left side in Fig. \ref{fig:GraphProof}, in which $I_{l_1}$ is connected to $O_{m_1}$ through $FF_{a_1},\cdots,FF_{a_k}$, which we name an ``earlier path.''
    Subsequently, we focus on the path towards $O_{m_2}$ as the right-hand side path in Fig. \ref{fig:GraphProof}, and we observe that the latest FF $FF_{a_k}$  is connected to $O_{m_2}$ via the combination node because of the assumption.
    If the attributes of $FF_{a_k}$ are all held after the earlier path, the path towards $O_{m_2}$ is connected to the earlier path and $(\rho_{m_2})_{l1}$ becomes $I$; thus, a latch attribute exists at some point.
Next, $FF_{a_{k-1}}$ is connected to $FF_{a_k}$ and we replicate the same argument.
Finally, $FF_{a_1}$ is connected to input $I$ at some point after $l_1$, which contradicts this assumption.
\end{proof}

We can now state the theorem of a universal property for multiple clock domain circuits.

\begin{theorem*}
The referring forms for multiple clock domain circuits (as shown in Fig. \ref{fig:DiagMCD}) are time preserving.
\label{th:the}
\end{theorem*}
\begin{proof}
Let $R$ be the set of referring forms.
We aim to define the order on $R(\mathbb{N}_0)$ as the order derived by each $\rho\in R$, but we must confirm its well-definedness by contradiction.
Assuming that antisymmetry does not hold, we have $\rho,\rho' \in R$ and $m_1,m_2,l_1,l_2 \in \mathbb{N}_0$ s.t.
\[
m_1\le m_2,\;l_1\le l_2,\; \rho_{m_1}=\rho'_{l_2}, \; \rho_{m_2} = \rho'_{l_1},\; \rho_{m_1}\neq\rho_{m_2}.
\]
Focusing on the latest $I$ appearance in $\rho_{m_1}$, according to the Lemma, $I$ appears later in $\rho_{m_2}$.
As $\rho_{m_2}=\rho'_{l_1}$, another later $I$ appears in $\rho'_{l_2}$, and from $\rho'_{l_2}=\rho_{m_1}$, those $I$ appearances occur at the same time.
For the second-latest $I$ appearance, we can reprise the same argument, and finally, we obtain $\rho_{m_1}=\rho_{m_2}$, which contradicts the assumption.
\end{proof}

Finally, we present a non-time-preserving example.
Fig. \ref{fig:2mem}(a) is a circuit with two data inputs $I_1$ and $I_2$, three control inputs $clock 1, clock 2$, and $sel$, and two memories that can be read selectively.
Two referring forms, $\rho$ and $\rho'$, are shown in Fig. \ref{fig:2mem}(b) and (c), respectively, and assuming that both undescribed pasts are the same, we cannot provide the order between $\rho_n=\rho'_{n+7}$ and $\rho_{n+3}=\rho'_{n+3}$ on the united image of the referring forms.

\section{Conclusion}

We have introduced referring forms as a novel approach for investigating the functional properties of sequential circuits.
Referring forms capture the transition behavior of circuits, specifically how past inputs are used to generate current outputs.
We defined the concept of time preservation for referring forms and demonstrated that multiple clock domain circuits exhibit this property.
This study provides a behavioral perspective on sequential circuits and offers new insights into their analysis and design.
We hope that future research will explore further theoretical developments and practical applications of referring forms in digital circuits.

\bibliographystyle{ACM-Reference-Format}
\bibliography{the}

\end{document}